%% file: disc-x.tex
\documentclass[pdflatex,sn-mathphys]{sn-jnl}% Math and Physical Sciences Reference Style
\usepackage{amsmath}
\usepackage{amsthm}
\usepackage{mathtools}

\usepackage{listings}
\definecolor{key-color}{rgb}{0,0,0}

\lstdefinestyle{mystyle}
{
	basicstyle= {\ttfamily\large},
    keywordstyle = {\bf\color{key-color}},
    morekeywords = {for,forall,while,endwith,with,if,and,or,then,endfor,else,foreach,endif, end, function,do,return, in},
}

\definecolor{codegreen}{rgb}{0,0.6,0}
\definecolor{codegray}{rgb}{0.5,0.5,0.5}
\definecolor{codepurple}{rgb}{0.58,0,0.82}
\definecolor{backcolour}{rgb}{0.95,0.95,0.92}
 
\lstdefinestyle{mystyle2}{
    backgroundcolor=\color{backcolour},   
    commentstyle=\color{codegreen},
    keywordstyle=\color{magenta},
    numberstyle=\tiny\color{codegray},
    stringstyle=\color{codepurple},
    basicstyle=\footnotesize,
    breakatwhitespace=false,         
    breaklines=true,                 
    captionpos=b,                    
    keepspaces=true,                 
    numbers=left,                    
    numbersep=5pt,                  
    showspaces=false,                
    showstringspaces=false,
    showtabs=false,                  
    tabsize=2
}
 
\jyear{2021}%

\theoremstyle{thmstyleone}%
\newtheorem{theorem}{Theorem}%  meant for continuous numbers
\newtheorem{lemma}[theorem]{Lemma}% 

\theoremstyle{thmstyletwo}%
\newtheorem{example}{Example}%

\theoremstyle{thmstylethree}%
\newtheorem{definition}{Definition}%

\input{defs}

\begin{document}

\title[Connected Components and Disjunctive Existential Rules]{Connected Components and Disjunctive Existential Rules}

%%=============================================================%%
%% Prefix	-> \pfx{Dr}
%% GivenName	-> \fnm{Joergen W.}
%% Particle	-> \spfx{van der} -> surname prefix
%% FamilyName	-> \sur{Ploeg}
%% Suffix	-> \sfx{IV}
%% NatureName	-> \tanm{Poet Laureate} -> Title after name
%% Degrees	-> \dgr{MSc, PhD}
%% \author*[1,2]{\pfx{Dr} \fnm{Joergen W.} \spfx{van der} \sur{Ploeg} \sfx{IV} \tanm{Poet Laureate} 
%%                 \dgr{MSc, PhD}}\email{iauthor@gmail.com}
%%=============================================================%%

\author*[1]{\fnm{Enrique} \sur{Matos Alfonso}}\email{gardero@image.ece.ntua.gr}

\author[1]{\fnm{Giorgos} \sur{Stamou}}
\email{gstam@cs.ntua.gr}

\affil[1]{\orgdiv{Division of Computer Science}, \orgname{National Technical University of Athens}, \orgaddress{\city{Zografou}, \postcode{15780}, \state{Attica}, \country{Greece}}}

%%==================================%%
%% sample for unstructured abstract %%
%%==================================%%

\abstract{In this paper, we explore conjunctive query rewriting, focusing on queries containing universally quantified negation within the framework of disjunctive existential rules. %Conjunctive query rewriting, a foundational technique for query answering with rules, involves backward application of rules to generate a union of conjunctive queries (\UCQ{}) as a rewriting of the original query, aiming for completeness without the need for rules. 
We address the undecidability of the existence of a finite and complete \UCQ{}-rewriting and the identification of finite unification sets (\fus{}) of rules. We introduce new rule classes, \emph{connected linear rules} and \emph{connected domain restricted rules}, that exhibit the \fus{} property for existential rules. Additionally, we propose \emph{disconnected disjunction} for disjunctive existential rules to achieve the \fus{} property when we extend the introduced rule fragments to disjunctive existential rules. We present \ecompleto{}, a system for efficient query rewriting with disjunctive existential rules, capable of handling \NUCQ{}s with universally quantified negation. Our experiments demonstrate \ecompleto{}'s consistent ability to produce finite \UCQ{}-rewritings, and describe the performance on different ontologies and queries.}


%%\pacs[JEL Classification]{D8, H51}

%%\pacs[MSC Classification]{35A01, 65L10, 65L12, 65L20, 65L70}

\maketitle

\section{Introduction}\label{intro}

Conjunctive query rewriting  \cite{thore, other_disj} is one main approaches to perform query answering in the presence of rules. Given a set of rules $\CalR$ and a query $\CalQ$, the rules are applied in a backward manner to generate a union of conjunctive queries (\UCQ{}), which is a rewriting of $\CalQ$. The aim of the process is to reach a \emph{complete} rewriting of $\CalQ$ that does not need the rules in order to represent all the answers of $q$. 
%\begin{example}   
%\end{example}

Conjunctive query rewriting is very convenient in scenarios where the system's data undergoes frequent modifications or expansions, and the queries and rules remain relatively stable.  Conversely, in situations where the data remains static, and different queries are to be responded to, a forward chaining approach \cite{forward} is more suitable for the query answering task.

In this paper, we focus on conjunctive query answering for queries that might contain universally quantified negation \cite{queries, queries2, db} with respect to disjunctive existential rules \cite{semantic, other_disj}. The existence of a finite and complete \UCQ{}-rewriting of an arbitrary query with respect to an arbitrary set of rules is an undecidable problem \cite{walking}. A set of rules that ensures the existence of a finite and complete \UCQ{}-rewriting with respect to any \UCQ{} is called \emph{finite unification set} (\fus{}). We also refer to a \fus{} as a \emph{rewritable} (\emph{first-order rewritable}) set of rules. Determining whether a set of rules is a \fus{} is also an undecidable problem. However, the authors in \cite{walking} provide a compilation of some properties that ensure the \fus{} property for the framework of existential rules. In the case of disjunctive existential rules, the only known way to obtain a \fus{} is expanding a \fus{} of existential rules with disjunctive existential rules that are disconnected, i.e., there are no variables shared between the hypothesis and the consequence of the rule.

In the present study, we introduce \emph{connected linear rules} and \emph{connected domain restricted rules}, two classes of rules that exhibit the \fus{} property for existential rules despite being more expressive than the classes they extend (linear rules and domain restricted rules respectively). The introduced classes do not ensure the \fus{} property for disjunctive existential rules. However, we introduce the concept of \emph{disconnected disjunction} for disjunctive existential rules that ensures those rule fragments to be a \fus{} even in the case of disjunctive existential rules.

In terms of practical implementation, we describe the latest version of our system, \ecompleto{}, which is designed to perform query rewriting with respect to disjunctive existential rules. This system is capable of handling \NUCQ{}s with universally quantified negation and offers an extension of \dlgpp{} for specifying disjunctive existential rules and negated atoms in queries.

To evaluate the performance of our system, we conducted experiments using two known ontologies, Lehigh University Benchmark \lubm{} \cite{lubm} and \travel{}, both enriched with additional axioms. We generate 500 \NUCQ{}s for each ontology and observed that \ecompleto{} successfully produce finite \UCQ{}-rewritings for all queries. Notably, the rewriting process for the \travel{} ontology is significantly faster and consumes less memory compared to the \lubm{} ontology.

\paragraph{Related Work}

To the best of our knowledge, the only work related to \UCQ{}-rewritings with respect to disjunctive existential rules was published by Lecl\'ere et al. \cite{other_disj}. The authors proposed a rewriting definition for disjunctive existential rules that is different to the one proposed in \cite{semantic}. Their definition uses piece rewritings to eliminate all the disjunctive components of a disjunctive existential rule and produce a conjunctive query. The main difference between their method and the one proposed in \cite{semantic} is that intermediate rules which may not lead to a \CQ{}, are avoided by them. Based on their findings, we could avoid expanding a disjunctive rule if one of its disjoints cannot be used to produce \UCQ{}-rewritings. The authors did not reference any implementations of their approach.
\paragraph{Paper Structure}
Section \ref{prelim} provides background concepts needed to understand the rest of the paper. Subsection \ref{setformulas} introduces concepts related to first-order logic formulas and how the elements of a formula are connected. Subsection \ref{er} presents the disjunctive existential rules framework. 
Section \ref{rew} focuses on rewritability. Subsection \ref{rew-disj} introduces the definition of \UCQ{}-rewriting for disjunctive existential rules and a backward chaining rewriting algorithm for disjunctive existential rules. Subsection \ref{fragments} mentions the rewritable existential rule fragments defined in the literature. Subsection \ref{expand} proposes two new existential rule classes that have the \fus{} property. It also introduces further restrictions in order to retain the \fus{} property when those rule classes are extended to disjunctive existential rules. 
Section \ref{eval} describes a new implementation and provides an experimental evaluation of the proposed rewriting algorithm. Section \ref{discussion} discusses the experimental evaluation. Finally, Section \ref{conclusion} concludes the paper and documents our most significant contributions. 
\section{Preliminaries}\label{prelim}

\subsection{Set Formulas and Connected Components}\label{setformulas}
%The reader is expected to be familiar with the standard definition of first-order logic formulas. In this paper, we restrict ourselves to FOL formulas without function symbols over a finite set of predicate names and a finite set of constant symbols. 
%We use the standard definitions for the entailment and equivalence of formulas. 

In this section, we provide essential background information which is necessary to understand the concepts and terminology used in this paper. We assume the reader is familiar with standard first-order logic (FOL) formulas. The reader is referred to \cite{Nienhuys-Cheng1997-1} in case a background reading is needed. 

We work exclusively with FOL formulas that lack function symbols and are built upon finite sets of predicates and constant symbols. Standard definitions for the entailment and equivalence of formulas are applied throughout this paper.

%The concepts of  conjunctive (disjunctive) set formulas (\CSF{}s and \DSF{}s) \cite{semantic} are introduced to help in the compact representation of disjunctive rules. However, both notations are often used in FOL when we write rules and clauses.

We introduce the concepts of conjunctive set formulas (\CSF{}s) and disjunctive set formulas (\DSF{}s) \cite{semantic} as tools for representing disjunctive rules in a compact way. However, both notations are often used in FOL to represent rules and clauses.

%A \emph{conjunctive (disjunctive) set formula} (\CSF{} and \DSF{}, respectively) is a set of formulas $\{F_1,\ldots,F_n\}$, interpreted as a conjunction (disjunction) of the formulas in the set.
%Let $\{F_1,\ldots,F_n\}$ be a set of formulas. A \CSF{} containing these formulas is written as $F_1,\ldots,F_n$. 
%On the other hand, a \DSF{} is written as $[F_1,\ldots,F_n]$.  An empty \CSF{} is equivalent to $\top$, and an empty \DSF{} is equivalent to $\bot$.
%We use parenthesis to avoid ambiguity when necessary, e.g., $[(A,B), D]$ is equivalent to $(A\wedge B)\vee D$.

A \emph{conjunctive set formula} (\CSF{}) is a set of formulas denoted as $F_1, \ldots, F_n$ and interpreted as the conjunction of these formulas $F_i$. A \emph{disjunctive set formula} (\DSF{}) is represented as $[F_1, \ldots, F_n]$ and interpreted as the disjunction of its formulas $F_i$. An empty \CSF{} is equivalent to $\top$, and an empty \DSF{} is equivalent to $\bot$.
Parentheses are used when necessary to avoid ambiguity, e.g., $[(A, B), D]$ is equivalent to $(A \land B) \lor D$.

%Axioms $A_i$ of an entailment operation are represented as a \CSF{} and a \DSF{} is used to represent consequences $Q_j$:\[ A_1,\ldots,A_n \models [Q_1,\ldots,Q_m].\]

In the context of the entailment operation, we represent the axioms $A_i$ using \CSF{}s and the consequences $Q_j$ using \DSF{}s, i.e.,
\[ A_1, \ldots, A_n \models [Q_1, \ldots, Q_m]. \]

%A \emph{term} is a constant or a variable. 
%An \emph{atom} is a formula $a(t_1,\ldots,t_n)$ where $a$ is a \emph{predicate} of \emph{arity} $n$. The arguments $t_i$ of the atom are terms. A \emph{literal} is an atom or a negated atom. Let $A$ be an atom. The \emph{complement} $\comp{l}$ of a literal $l=A$ is $\neg A$. The \emph{complement} $\comp{l}$ of a literal $l=\neg A$ is $A$. A literal $l=A$ is \emph{positive} (or of \emph{positive polarity}). On the other hand, a literal $l=\neg A$ is \emph{negative} (or of \emph{negative polarity}). Two literals are \emph{complementary} if one is the complement of the other.

A \emph{term} can be a constant or a variable.
An \emph{atom} is a formula of the form $a(t_1, \ldots, t_n)$, where $a$ is a \emph{predicate} with \emph{arity} $n$, and the arguments $t_i$ are terms.
A \emph{literal} is either an atom or its negation. A literal $l = A$ is \emph{positive}, while $l = \neg A$ is \emph{negative}.
The \emph{complement} $\comp{l}$ of a literal $l$ is the literal with the opposite sign, i.e., $\comp{A} = \neg A$ and $\comp{\neg A} = A$. Two literals $l_1$, $l_2$ are \emph{complementary} if one is the complement of the other $\comp{l_1}=l_2$.

%The expression $\atomx{vars}{F}$ denotes the set of all the variables that appear in $F$. The variables of a formula can be \emph{universally quantified}, \emph{existentially quantified}, or \emph{free}.  A \emph{closed} formula contains no free variable. A \emph{ground} formula contains no variable.

The expression $\atomx{vars}{F}$ denotes the set of variables appearing in a formula $F$.
In a formula,  variables can be \emph{universally quantified}, \emph{existentially quantified}, or \emph{free}.
A formula is \emph{closed} when it contains no free variables. A formula is \emph{ground} if it contains no variables at all.

A \emph{substitution} $\theta = \{X_1\leftarrow t_1,\ldots X_n\leftarrow t_n\}$ is a finite mapping of variables $X_i$ to terms $t_i$. %The expression $F\theta$ is the result of applying a substitution $\theta$ on an expression $F$ and it is obtained by replacing in $F$ every occurrence of every variable $X_i$ by the term $t_i$. 
The result of applying a substitution $\theta$ to a formula $F$ is denoted as $F\theta$ and it is obtained by replacing in $F$ every occurrence of every variable $X_i$ by the corresponding term $t_i$. 

%Let $F$ be an expression and $\theta$ the substitution  $\{X_1\leftarrow Y_1,\ldots X_n\leftarrow Y_n\}$. We say $\theta$ is a \emph{renaming substitution} for $F$, if each $X_i$ occurs in $F$, and $Y_1,\ldots, Y_n$ are distinct variables such that each $Y_i$ is either equal to some $X_j$ in $\theta$, or $Y_i$ does not occur in $F$.

A substitution $\theta$ is a \emph{renaming substitution} for the expression $F$, if each $X_i$ occurs in $F$, the variables $Y_1,\ldots, Y_n$ are distinct and each $Y_i$ is either equal to some $X_j$ in $\theta$, or $Y_i$ does not occur in $F$.

The \emph{composition} of two substitutions, $\theta$ and $\sigma$, is a new substitution $\theta\sigma$ such that $F\theta\sigma=(F\theta)\sigma$ for any expression $F$. A substitution $\theta$ is \emph{more general than} another substitution $\sigma$ if there exists a substitution $\gamma$ such that $\sigma=\theta\gamma$.

A \emph{unifier} for a set of expressions $S=\{F_1 ,\ldots,F_n\}$ is a substitution $\theta$ such that, $F_1\theta=F_2\theta,\ldots,F_{n-1}\theta=F_n\theta$. The expressions in $S$ are \emph{unifiable}, if there is a unifier for $S$. 
The most general unifier (\emph{mgu}) of $S$ is a unifier that is more general than any other unifier for $S$ and is unique up to variable renaming.
%A unifier for $S$ that is more general than any other unifier of $S$ is the \emph{most general unifier} (\emph{mgu}) of $S$ and it is denoted by $mgu(S)$. Different \emph{mgu}s for the same set $S$ are unique up to variable renaming.

A \emph{hypergraph} is a pair $\left\langle A, E\right\rangle$ of nodes $A$ and \emph{hyperedges} $E$, where the hyperedges are non-empty subsets of $A$. We represent a \CSF{} of atoms $F$ by a hypergraph $\left\langle A, E\right\rangle$, where $A=\atomx{vars}{F}$.  Each atom $A_i$ in $F$ represents a hyperedge $\atomx{vars}{A_i}\in E$ that connects its variables. Following, we introduce some properties for \CSF{}s of atoms that follow from the defined hypergraph representation.

Let $F$ be a \CSF{} of atoms. The  \emph{cardinality} of a formula $F$ is the number of variables in the formula, i.e., $\atomx{card}{F}=\vert\atomx{vars}{F}\vert$. The \emph{width} of a formula $F$ (denoted by $\atomx{width}{F}$) is the number of atoms that have at least one variable in their arguments.
Two variables $u$ and $v$ in $\atomx{vars}{F}$ are \emph{connected} iff they both belong to the some hyperedge ($\exists{A} \in F \vert \{v,u\} \subseteq \atomx{vars}{A}$), or iff there is another variable $z$ in $F$ that is connected to both $u$ and $v$. 

A \CSF{} of atoms $F$ is \emph{connected} iff all the atoms in it contain variables and all the variables are connected to each other.  An atom that has only constants in its arguments is a connected formula by itself and it is represented by an empty hypergraph. The constants in the formula play no role in their hypergraph representation. 

A \CSF{} $F$ can be partitioned into a set $\{U_1,\ldots,U_n\}$ of connected \CSF{}s such that if $v\in \atomx{vars}{U_i}$ is connected to $u\in \atomx{vars}{U_j}$, then $i=j$. The formula $F$ is equivalent to a \CSF{} of its connected components, i.e., $F=U_1,\ldots,U_n$.

The \emph{connected cardinality} and \emph{connected width} of $F$ (denoted by $\ccard{F}$ and $\cwidth{F}$) is defined as the maximum cardinality and width among its connected components, respectively. In other words, $\ccard{F}=\max_i{(\atomx{card}{U_i})}$ and $\cwidth{F}=\max_i{(\atomx{width}{U_i})}$. 

The connected cardinality and connected width of a \DSF{} $F=[F_1,\ldots,F_m]$ is the maximum connected cardinality and width of the formulas $F_i$, respectively, i.e., $\ccard{F}=\max_i{(\ccard{F_i})}$ and $\cwidth{F}=\max_i{(\cwidth{F_i})}$.

\begin{lemma}
\label{partition_sep}
Let $G$ be a \CSF{} and let $\{U_1,\ldots,U_n\}$ be the partition of a given  \CSF{} $F$ of atoms into connected \CSF{}s $U_i$.  Then,
  \[
  \begin{array}{rl}
      G \models F ~~\text{iff}~~ &  G \models U_i ~\text{for every }~ U_i.\\
       
  \end{array}
  \
  \]
 
\end{lemma}
\begin{proof}
Given that no variables are shared between the connected components $U_i$, we can safely combine the assignments for the variables within each $U_i$ without introducing conflicts in the values assigned to each variable. 
Detailed proof for this lemma is given by Tessaris \cite{tessaris}.  
\end{proof}

For \CSF{}s with a bounded connected cardinality or the connected width, we can ensure the existence of a finite number of equivalence classes for these formulas.
\begin{lemma}
\label{finite_cardk}
Let $k$ be a natural number. There is a finite number of equivalence classes of \CSF{}s of atoms with connected cardinality of at most $k$ that can be constructed using a finite set of predicates and constants.

\end{lemma}
\begin{proof}
The equivalence of two \CSF{}s composed of atoms holds if and only if they can be unified through a renaming substitution.
%Since we have finitely many predicates and constant symbols, and at most $k$ different variables, we can combine them in a finite number of ways $M$ to form a connected \CSF{}. 
 Given that we have a finite set of predicates and constant symbols, and a maximum of $k$ distinct variables, we can construct \CSF{}s in a finite number of ways, denoted as M, to ensure they remain connected.
In a \CSF{} consisting of more than $M$ connected \CSF{}s we know that some of the connected components are renamings of others, and keeping only one of them yields a formula equivalent to $F$ (according to Lemma \ref{partition_sep}). Hence, there are at most $2^M$ different equivalence classes for \CSF{}s with a maximum of $k$ distinct variables.
\end{proof}

\begin{lemma}
\label{finite_cardk_edges}
Let $k$ be a natural number. There are a finite number of equivalence classes of \CSF{}s of atoms with connected widths of at most $k$ that can be constructed using a finite set of predicates and constants.

\end{lemma}
\begin{proof}
This follows directly from Lemma \ref{finite_cardk} because an upper bound for the connected width of a \CSF{} implies that there is also an upper bound for the connected cardinality for a set of finitely many predicates.
\end{proof}

\subsection{Disjunctive Existential Rules Framework}\label{er}

A \emph{conjunctive query} (\CQ{}) is a \CSF{} $l_1,\ldots,l_n$ of positive literals (atoms) $l_i$ where all the variables $\vvector{X} = \atomx{vars}{ l_1,\ldots,l_n }$ are existentially quantified, i.e., an expression of the form $\exists{\vvector{X}}~ l_1,\ldots,l_n$.
Queries that permit negation in the literals $l_i$ are referred to as \emph{conjunctive queries with negation} (\NCQ{}). All the variables $\vvector{X}$ that appear in the positive literals of a \NCQ{} are existentially quantified. In this paper, we make use of \emph{universally quantified negation} \cite{queries,db, queries2}, i.e., all the  variables $\vvector{Z}$ that appear only in  negative literals  are universally quantified: $\exists{\vvector{X}}\forall{\vvector{Z}}~ l_1,\ldots,l_n.$. Universal quantification ensures that our queries are safe and domain-independent. 
 From now on, we exclude quantifiers in queries, given that the rules for writing them are well-established.
 The set of variables that appear in both positive and negative literals is called the \emph{frontier} of the query.
 %\todo{introduce answer variables and talk about them}
%Note that for now, we do not introduce the concept of \emph{answer variables}. 
The queries we define are commonly referred to as \emph{Boolean conjunctive queries}. Throughout the paper, by conjunctive query, we mean Boolean conjunctive query. It is important to mention that the algorithms introduced in this paper can also be modified to accommodate queries with answer variables, as detailed in \cite{semantic}.

A \DSF{} of conjunctive queries (conjunctive queries with negation) is referred to as a \emph{union of conjunctive queries} (\UCQ{}) (\emph{union of conjunctive queries with negation} (\NUCQ{})). For a \NUCQ{} $\CalQ$, the set of \NCQ{}s in $\CalQ$ with exactly $k$ negated atoms is denoted by $\CalQ^{\neg k}$, and the set of \NCQ{}s in $\CalQ$ with two or more negated atoms is denoted by $\CalQ^{\neg \#}$. We use the term \emph{query} to refer to either a \CQ{}, \NCQ{}, \UCQ{} or \NUCQ{}.

A \emph{fact} is a \CSF{} $a_1,\ldots,a_n$ of atoms $a_i$, where all variables are assumed to be existentially quantified. Note that we omit the explicit use of existential quantifiers.
%The definition of facts is equivalent to the definition Boolean conjunctive query. However, facts are used to express existing knowledge, while queries represent formulas to be entailed and they have different roles in the process of reasoning.
The concept of facts aligns with the definition of Boolean conjunctive queries. Nevertheless, it is essential to distinguish their roles: facts are employed to convey existing knowledge, whereas queries serve as formulas to be validated, having distinct roles in the reasoning process.

A \emph{rule} is a closed formula  of the form
\[\forall{\vvector{X}}\, \exists{\vvector{Y}} ~ B \rightarrow H,\]
where the \emph{body} $B$ is a \CSF{} of atoms, and the \emph{head} $H$ is a \DSF{} in which all $H'\in H$ are \CSF{}s of atoms. 
%The set $\vvector{X}=\atomx{vars}{B}$ contains the variables that appear in the body, and they are universally quantified. On the other hand, $\vvector{Y}=\atomx{vars}{H}\setminus \atomx{vars}{B}$ are the variables that appear only in the head. They are existentially quantified, and they are called \emph{existential variables}. The \emph{frontier} of a rule is the set of variables that are present in both the body and head of the rule: $\atomx{vars}{B}\cap\atomx{vars}{H}$. We omit quantifiers when writing a rule.
The set $\vvector{X}=\atomx{vars}{B}$ includes the variables which occur in the body and are universally quantified. On the other hand, $\vvector{Y}=\atomx{vars}{H}\setminus \atomx{vars}{B}$ are the variables exclusively appearing in the head of the rule. These variables are existentially quantified and are referred to as \emph{existential variables}. The \emph{frontier} of a rule refers to the variables present in both the body and head of the rule, denoted as $\atomx{vars}{B}\cap\atomx{vars}{H}$. We simplify the notation by omitting quantifiers when expressing a rule.

A \emph{disjunctive existential rule} is a rule with more than one disjoint element in the head, i.e., $\size{H}>1$. In contrast, an \emph{existential rule} is a rule featuring exactly one disjoint element in the head. For simplicity, we represent the head of the existential rule as a \CSF{} of atoms.  
A \emph{negative constraint} is a rule with an empty disjoint in the head, i.e., $B\rightarrow\bot$. In contexts where it is evident that we are referring to a negative constraint, we may omit the ``$\rightarrow\bot$''. %Sometimes we also refer to a negative constraint as a constraint.
Occasionally, we also use the terms constraint and negative constraint interchangeably.

%A \CSF{} of atoms $Q$ \emph{depends} on a rule $r$ iff there is a \CSF{} of atoms $F$ such that $F\nmodels Q$ and $F, r\models Q$. A rule $r_i$ depends on a rule $r_j$ iff the body of $r_i$ depends on $r_j$. The concept of rule dependencies allows us to define the \emph{graph of rule dependencies} (GRD) \cite{aGRD}, which is a graph where nodes are rules, and a directed edge between two nodes represents the existence of a dependency between the corresponding rules. 

%A \emph{knowledge base} (\KB{}) $\CalK=\left\langle\CalR,\CalD\right\rangle$ is composed by a \CSF{} $\CalR$ of rules and a \CSF{} of facts $\CalD$. 
%For a set of rules $\Rules$, the set of constraints in $\Rules$ is $\Rules^\bot$, the set of existential rules is $\Rules^\exists{}$ and the set of disjunctive existential rules is $\Rules^\vee{}$. 
In the context of a rule set $\Rules$, we denote the set of constraints within $\Rules$ as $\Rules^\bot$, the set of existential rules as $\Rules^\exists{}$ and the set of disjunctive existential rules as $\Rules^\vee{}$. A set of facts is often denoted as $\CalD$.
%A knowledge base  $\left\langle\CalR,\CalD\right\rangle$  is a \emph{disjunctive} knowledge base (\DKB{}) if $\Rules^\vee\neq\emptyset$, otherwise it is an \emph{existential} knowledge base (\EKB{}).

In this paper, we study the \emph{query entailment} problem with respect to disjunctive existential rules, i.e., 
\begin{equation}
\label{entailment}
    \CalR, \CalD \entailsp \CalQ. 
\end{equation}

We address the problem \eqref{entailment} by transforming it into the entailment of a \UCQ{} $\CalQ'$ with respect to the facts $\CalD$, i.e.,
\[
\CalD\entailsp \CalQ'.
\]

A \UCQ{} $\CalQ'$ is a \emph{\UCQ{}-rewriting} of $Q$ with respect to $\CalR$ if for any $\CalD$ when the following condition holds:
\begin{equation}
\label{rweq}
\CalD \models \CalQ' ~~\text{implies}~~\CalR, \CalD \models \CalQ.
\end{equation}
Every \CQ{}s in $\CalQ'$ is a \CQ{}-\emph{rewriting} of $Q$ with respect to $\CalR$.
If the converse of \eqref{rweq}
\[
\CalR, \CalD \models \CalQ~~\text{implies}~~\CalD \models \CalQ'
\]
is also true for any $\CalD$, then $\CalQ'$ is a \emph{complete} \UCQ{}-rewriting of $\CalQ$ with respect to $\CalR$, i.e.,
\[
\CalR, \CalD \models \CalQ~~\text{iff}~~\CalD \models \CalQ'
\]
%Note that according to our definition, a \UCQ{}-rewriting may not be complete. In this respect, our definition follows the definition of \UCQ{}-rewriting  from \cite{thore} because we extend many of the concepts and algorithms proposed by the authors.

The negative constraints in the entailment problem can be transformed into queries, i.e.,
\[\existential{\CalR},\disjunctive{\CalR}, \constraints{\Rules}, \CalD \models \CalQ ~~~\text{iff}~~~ \CalR,\CalR^\vee, \CalD \models \neg \constraints{\Rules}, \CalQ.\]
Additionally, the conjunctive queries with negation, $\CalQ^{\neg 1}$ and $\CalQ^{\neg \#}$ can be transformed into existential rules and disjunctive existential rules respectively:
\begin{equation}
\label{reduction}
\begin{array}{l}
    \existential{\CalR},\disjunctive{\CalR}, \constraints{\CalR},\CalD \models \CalQ ~~~\text{iff}~~~  \\
     \phantom{x}(\existential{\CalR}, \neg\CalQ^{\neg 1}),(\disjunctive{\CalR}, \neg\CalQ^{\neg \#}), \CalD \models \neg \constraints{\CalR}, \CalQ^{\neg 0}, 
\end{array}
\end{equation}

Consequently,  for the remainder of this paper, we focus on the problem of finding a \UCQ{}-rewriting of an input \UCQ{}s with respect to (disjunctive) existential rules without constraints.  

\section{Rewritability}\label{rew}

%We focus on class of rules that ensure the existence of a finite and complete  \UCQ{}-rewriting for every input query.

In this Section, we begin by presenting an algorithm for finding \UCQ{}-rewritings introduced in \cite{semantic}. The algorithm relies on general rewriting steps utilizing disjunctive rules and conjunctive queries to generate rules with fewer disjoints and eventually other existential rules. After, we introduce the existing rewritable rule fragments defined in the literature. We also present two new existential rule fragments that also have the \fus{} property. Additionally, we introduce restrictions that ensure the \fus{}  property of those rule fragments for the case of disjunctive rules.

\subsection{Rewriting Steps and Algorithms for Disjunctive Existential Rules}\label{rew-disj}

A  \emph{piece-based} rewriting step, as defined within the existential rules framework \cite{thore}, corresponds to the rewriting step defined in \cite{semantic}.

\begin{definition}[Rewriting Step]
\label{rewrite}
Let $r =  B\rightarrow H$ be an existential rule, and $Q$ a conjunctive query. If there is a subset $H'\subseteq H$ that unifies with some $Q'\subseteq Q$ through a mgu $\theta$  (i.e., $H'\theta=Q'\theta$) such that
\begin{enumerate}
    \item if $v \in \atomx{vars}{Q\setminus Q'}$ and $v\neq v\theta$, then $v\theta$ is a frontier variable of $r$ or a constant, and
    \item if $v$ is an existential variable of the rule $r$, then $v\theta\notin  \atomx{vars}{Q\setminus Q'}$, 
\end{enumerate}
then the query $(B \cup (Q\setminus Q')) \theta$ is a \emph{rewriting} of $Q$ using the existential rule $r$.
\end{definition}

The authors in \cite{semantic}  define a corresponding rewriting step for a disjunctive existential rule and a \CQ{}. It is a generalization of  Definition \ref{rewrite} with the goal to support both existential rules and disjunctive existential rules.

\begin{definition}[General (Disjunctive) Rewriting Step]
\label{drewrite}
Let $r = B\rightarrow H$ be a rule, and $Q$ a conjunctive query. If there is a subset $H'\subseteq H$, and for each $h_i \in H'$ there is a subset $h'_i\subseteq h_i$ that unifies with a $Q' \subseteq Q$ through a mgu $\theta$ (i.e., $h'_1\theta = \ldots h'_n\theta = Q'\theta$) such that

\begin{enumerate}
    \item if $v \in \atomx{vars}{Q\setminus Q'}$ , then $v\theta$ is a frontier variable of $r$ or a constant, and
    \item if $v$ is an existential variable of the rule $r$, then $v\theta\notin  \atomx{vars}{Q\setminus Q'}$,
\end{enumerate}

then $(B \cup  (Q\setminus Q') \rightarrow H \setminus H') \theta$ is a \emph{rewriting} of $Q$ using the rule $r$. A rewriting step is a \emph{disjunctive rewriting step} if the rule used is a disjunctive existential rule. 
\end{definition}

A disjunctive rewriting step results in either a disjunctive rule with fewer disjunctive components, an existential rule in case $\size{(H \setminus H')\theta}=1$ or a negative constraint (the negation of a conjunctive query) in case $H=H'$.

\begin{example}
Consider the following disjunctive existential rule:
\[
\begin{array}{rl}
     r_1=\atomx{diabetesRisk}{X} \rightarrow 
     & [(\atomx{diabetic}{Y}, \\ 
     & \phantom{xxxx}\atomx{sibling}{Y,X}),\\
     & \phantom{(}(\atomx{diabetic}{Z},\\
     & \phantom{xxxx}\atomx{parent}{Z,X})]. \\
\end{array}
\]
If we want to rewrite the query $Q=\atomx{diabetic}{X_1}$, to verify the presence of people with diabetes, we can obtain the \UCQ{}-rewriting $[\atomx{diabetic}{X_1}, diabetesRisk(X)],$ using $r_1$ with the unifier $\theta=\{Y\xleftarrow{} X_1, Z\xleftarrow{} X_1\}$.

Alternatively, if together with $r_1$ we have a query corresponding to a negative constraint $q_c=\atomx{singleChild}{X_1},\atomx{sibling}{Y_1,X_1}$ and the conjunctive query $Q'=diabetic(Y_2),parent(Y_2,X_2)$ verifying the existence of a diabetic parent, we obtain the following existential rule
\[
\begin{array}{rl}
     \atomx{diabetesRisk}{X}, \atomx{singleChild}{X} \rightarrow & \atomx{diabetic}{Z}, \\
     & \atomx{parent}{Z,X},
\end{array}
\]
as the result of rewriting $q_c$ using the rule $r_1$ and the unifier $\theta_2=\{X_1\leftarrow X, Y_1\leftarrow Y\}$.
Using the new existential rule we obtain the following \UCQ{}-rewriting: 
\begin{equation*}
    \begin{array}{l}
         [(\atomx{singleChild}{X},\atomx{sibling}{Y,X}), \\
         \phantom{[}(diabetic(Y),parent(Y,X)),\\
         \phantom{[}(\atomx{diabetesRisk}{X}, \atomx{singleChild}{X})].
 
    \end{array}
\end{equation*}
Note that the final \UCQ{}-rewriting also contains queries that correspond to constraints. They are possible reasons for which a query can be entailed, i.e., inconsistency in our facts.
\end{example}

Using the above-mentioned rewriting steps, we define \emph{rewriting} with respect to  disjunctive rules.

\begin{definition}[Rewriting]\label{rwdef}
Let $\left\langle\CalR,\CalQ\right\rangle$ be a tuple consisting of a set $\Rules$ of rules and a \UCQ{} $\CalQ$. A \emph{one-step rewriting} $\left\langle\CalR',\CalQ'\right\rangle$ of $\left\langle\CalR,\CalQ\right\rangle$ is obtained by adding to $\CalR$ or to $\CalQ$, as appropriate, the result $f'$ of a general rewriting step  that uses one of the conjunctive queries in $\CalQ$ and a rule in $\Rules$, i.e., $\CalQ' = \CalQ\cup (\neg f')$ if $f'$ is a negative constraint, $\CalQ' = \CalQ\cup f'$ if $f'$ is a conjunctive query, otherwise $\CalR' = \CalR\cup (f')$.

A \emph{$k$-step rewriting} of $\left\langle\CalR,\CalQ\right\rangle$ is obtained by applying a one-step rewriting to a $(k-1)$-step rewriting of $\left\langle\CalR,\CalQ\right\rangle$. For any $k$, a $k$-step rewriting of $\left\langle\CalR,\CalQ\right\rangle$ is a \emph{rewriting} of $\left\langle\CalR,\CalQ\right\rangle$. 

\end{definition}

A rewriting as defined in \ref{rwdef} is sound and complete.

\begin{theorem}[Soundness and Completeness of Rewritings]
Let $\left\langle \CalR,\CalD\right\rangle$ be a tuple consisting of a set $\Rules$ of rules and a \UCQ{} $\CalQ$. Then $\CalR,\CalD \models \CalQ$ iff there is a rewriting  $\left\langle\CalR',\CalQ'\right\rangle$ of $\left\langle \CalR,\CalQ\right\rangle$ such that  $\CalD \models Q_i$ for some conjunctive query $Q_i$ in $\CalQ'$.
\end{theorem}
\begin{proof}
The $k$-step rewriting of $\left\langle \CalR,\CalQ \right\rangle$ is based on a constraint derivation as defined in \cite{semantic}, i.e., resolution derivations that always use a clause with all its literals negated. Moreover, such a rewriting can be mapped to a constraint derivation.
Considering that constraint derivations are sound and complete \cite{semantic}, this theorem also holds. 
\end{proof}

\begin{algorithm}  \centering
\begin{lstlisting}[frame=single,mathescape,style=mystyle]
function rewrite$_k(\Rules,\CalQ)$
  do
    $\CalR_{old} := \CalR$
    $\CalQ_{old} := \CalQ$
    $\CalQ :=$ rewrite$_k^\exists(\existential{\CalR}, \CalQ)$
    $\Rules :=$ rewrite$^\vee(\Rules, \CalQ)$
  while $(\CalQ\neq\CalQ_{old}$ or $\CalR\neq\CalR_{old})$
  return $\CalQ$  
end function
\end{lstlisting}
\caption{Function to rewrite \UCQ{}s with respect to existential rules and disjunctive existential rules.}
\label{rwgen}
\end{algorithm}

Algorithm~\ref{rwgen} presents the function \texttt{rewrite}$_k$/2, which computes the rewritings of $\left\langle \CalR,\CalQ \right\rangle$, for a given a set of rules $\CalR$ and a \UCQ{} $Q$, and yields the corresponding \UCQ{}-rewriting component.
The algorithm alternates between computing the rewri\-tings of \CQ{}s using existential rules (function {\texttt{rewrite}$_k^{\exists}$/2}) and computing the rewritings using disjunctive existential rules (function \texttt{rewrite}$^\vee$/2). 
New \CQ{}s are used to produce more rules, and new existential rules are used to produce more \CQ{}s until a fixed point is reached, i.e., until no new rule or conjunctive query is produced.

All \CQ{}s generated by Algorithm~\ref{rwgen} are computed according to Definition \ref{rwdef}; this ensures the correctness of the rewritings produced, i.e., every \CQ{} that is generated is a \CQ{}-rewriting of the input query with respect to the input sets of rules and constraints.

A detailed description of the rewriting function for disjunctive existential rules (\texttt{rewrite}$^\vee$/2) can be found in \cite{semantic}. It generates all the possible rules using an input \UCQ{} rewriting. Due to the fact that new rules have less disjunctive components in the head, the output is always finite. Therefore, the completeness of the result of Algorithm \ref{rwgen} relies totally on the completeness of function {\texttt{rewrite}$_k^{\exists}$/2}. 

The function {\texttt{rewrite}$_k^{\exists}$/2} represents depth controlled version of a complete rewriting algorithm proposed on \cite{thore}. It implements a breath-first expansion process where each iteration expands a new level of conjunctive queries. The parameter $k$ that allows us to control how many levels of \CQ{}s will be expanded and ensures termination of each individual call for $k\neq \infty$. However, the loop in Algorithm \ref{rwgen} will keep on calling the function as long as new \CQ{}s are generated, without affecting the completeness of the whole rewriting process. The function yields only the most general \CQ{}-rewritings. 
If there is a finite and complete \UCQ{}-rewriting of the input \UCQ{}, the function \texttt{rewrite}$_k^\exists{}$/2 will find it within a finite number of calls. 

\subsection{Existing Rewritable Fragments}\label{fragments}

The termination of Algorithm \ref{rwgen} is studied in \cite{semantic}. The algorithm stops, if there is a finite and complete \UCQ{}-rewriting of the input query with respect to the rules.
\begin{theorem}
\label{finiterw-stops}
Let $\Rules$ be a set of rules and $\CalQ$ a \UCQ{}. If a \UCQ{} $\CalQ$ has a finite and complete \UCQ{}-rewriting  with respect to $\Rules$, then Algorithm~\ref{rwgen} stops for any finite value  of $k$. 
\end{theorem}
For the proof, we refer the reader to \cite{semantic}.

 The problem of knowing if there exists a finite \UCQ{}-rewriting for any \UCQ{} with respect to an arbitrary set of existential rules is undecidable \cite{walking}. A set of existential rules that ensures the existence of a finite \UCQ{}-rewriting for any \UCQ{} is called a finite unification set (\fus{}) \cite{classes}. We extend the concept of \fus{} to disjunctive existential rules.
 
 There are some classes of existential rules that have the \fus{} property:
\begin{enumerate}
    \item \emph{Linear} existential rules \cite{classes}: existential rules with one atom in the body.
    \item \emph{Disconnected} existential rules \cite{disc}: existential rules that do not share variables between the body and the head.
    \item \emph{Domain restricted} rules \cite{classes}: existential rules that each atom in the head contains none or all of the variables in the body. 
    \item \emph{Acyclic graph of rule dependencies} (\emph{aGRD}) \cite{aGRD}: existential rules that do not contain cycles in the \emph{graph of rule dependencies}.
    \item \emph{Sticky} rules \cite{sticky}: Each marked variable occurs at most once in a rule body. The marked variable set is built from a rule set using the following marking procedure: (i) for each rule $r_i$ and for each variable $v$ occurring in the body of $r_i$, if $v$ does not occur in all atoms of the head of $r_i$, mark (each occurrence of) $v$ in the body of $r_i$; (ii) apply until a fixpoint is reached: for each rule $r_i$, if a marked variable $v$ appears at position $p[k]$ in the body of $r_i$, then for each rule $r_j$ (including $i = j$) and for each variable $x$ appearing at position $p[k]$ in the head of $r_j$, mark each occurrence of $x$ in the body of $r_j$. 
    %\todo{check definition.}
\end{enumerate}

Extending the concept of linear rules to disjunctive existential rules is not sufficient to ensure the \fus{} property because it may only ensure the existence of a finite \UCQ{}-rewriting if the input query is an atomic query.

\begin{theorem}
\label{lin-stops}
Let $\Rules$ be a set of rules and $\CalQ$ a \UCQ{}. If $\CalQ$ contains only atomic queries and $\Rules$ only linear rules, then Algorithm \ref{rwgen} stops for any value of $k$.
\end{theorem}

If a set $\Rules$ of existential rules is a \fus{} and a set of the new existential rules generated by the function \texttt{rewrite}$^\vee$/2 is also a \fus{}, combining them could yield a new set of existential rules that is not a \fus{} \cite{walking}. Therefore, we need stronger conditions to ensure that we always call {\texttt{rewrite}$_k^{\exists}$/2} with a set of existential rules that is a \fus{}. In general, two \fus{} $\{\Rules_1,\Rules_2\}$ can be combined into a \fus{} if none of the rules of one set depends on the rules of the other set.

In \texttt{rewrite}$^\vee$/2, even if the resulting set of existential rules $\Rules$ is a \fus{}, the process of generating new rules could potentially continue forever after new \CQ{}s are generated. Therefore, we need ways to ensure that the total number of existential rules generated is bounded, i.e., there is a point beyond which the algorithm will not produce new rules.                       

Rules that do not share variables between the head and the body produce rewritings where the introduced body of the rule is not connected to the remaining part of the query.

A rule $B\rightarrow H$ with an empty frontier is a \emph{disconnected} rule, i.e., ${\atomx{vars}{B}\cap \atomx{vars}{H} = \emptyset}$. Disconnected rules can still share constants between the body and the head of the rule and this allows us to express knowledge about specific individuals.

\begin{theorem}
\label{disc-fus}
Let $\Rules_1$ be a \fus{} and $\Rules_2$ a set of disconnected existential rules. The union of both sets $\Rules_1\cup\Rules_2$ is also a \fus{}.
\end{theorem}

Theorem \ref{disc-fus} allows us to extend the \fus{} property to disjunctive existential rules. For a detailed proof check \cite{walking}.

\begin{theorem}
\label{disc-stops}
Let $\Rules$ be a set of rules. If $\existential{\Rules}$ is a \fus{}, and $\disjunctive{\Rules}$ a set of disconnected disjunctive existential rules, then $\Rules{}$ is also a \fus{}.
\end{theorem}
\begin{proof}
The theorem can be proven by ensuring that Algorithm \ref{rewrite} will always stop. The reader is referred to \cite{semantic} for a detailed proof.
\end{proof}

\subsection{Expanding the Existing Fragments}\label{expand}

Domain restricted (\dr{}) rules  \cite{classes} are existential rules where all the atoms in the head contain none or all of the variables in the body of the rule. However, if we consider rules where the bodies can have more than one connected component, then the definition of \dr{} rules can be generalized.

\begin{definition}[Connected domain restricted rule]
\label{cdrdef}

A rule is called \emph{connected domain restricted} (\cdr{}) rule if for every connected component $C$ in the body of the rule and for every atom $h$ in the head, $h$ contains none or all the variables of $C$. 

%A \cdr{} rule is a \kcdr{k} rule, if the atoms in the head of the rule does not contain variables from more than $k$ different connected components of the body of the rule.
\end{definition}

\begin{example}[Common ancestor and six degrees of separation rules]
\label{ex:cdr}
In biology and genealogy, the \emph{most recent common ancestor} (MRCA), \emph{last common ancestor} (LCA), or \emph{concestor} of a set of organisms is the most recent individual from which all the organisms of the set are descended. We could express a simpler rule stating that for every two organisms there exists a common ancestor:
\begin{equation*}
\atomx{organism}{X}, \atomx{organism}{Y} \rightarrow \atomx{organism}{Z}, \atomx{ancestor}{Z,X}, \atomx{ancestor}{Z,Y}
\end{equation*}
The rule is obviously not domain restricted but it is connected domain restricted.

Another example of \cdr{} rule that is not a \dr{} is the \emph{six degrees of separation} rule. It describes the idea that all people are six, or fewer, social connections away from each other.
\begin{equation*}
\begin{array}{l}
     \atomx{person}{X}, \atomx{person}{Y} \rightarrow \atomx{knows}{X,X_1}, \atomx{knows}{X_1,X_2}, \atomx{knows}{X_2,X_3}, \\
     \phantom{\atomx{person}{X}, \atomx{person}{Y} \rightarrow} \atomx{knows}{X_3,X_4}, \atomx{knows}{X_4,X_5}, \atomx{knows}{X_5,Y}   \\
\end{array}
\end{equation*}

In the example rules we assume that the predicate \pred{ancestor}{2} is irreflexive and antisymmetric, and \pred{knows}{2} is reflexive and symmetric.

\end{example}

Atoms in the head of a \cdr{} rule $r$ contain all the variables of some (possibly none) connected components in the body of the rule. We can be sure that the rewritings of a \CQ{} $q$ with respect to $r$ will not introduce new variables that are connected to the variables in the part of $q$ that is not modified by the rewriting. Some new variables might be introduced but they will be in isolated connected components. Hence, the connected cardinality of the rewritings with respect to \cdr{} rules is not increasing. Consequently, the class of \cdr{} rules also has the \fus{} property.

\begin{theorem}
\label{cdrfus}
A set of \cdr{} existential rules is a \fus{}.    
\end{theorem}
\begin{proof}
A \UCQ{}-rewriting $q'$ that is generated using a \cdr{} rule $r$ and a \CQ{} $q$ has new connected components $C'_i$ that are either (i) not connected to the rest of the query or (ii) that all their variables were already present in an atom of $q'$. Therefore, the only new variables (w.r.t. the variables in $q$) that are introduced in the rewritings are part of disconnected components that come from the body of the set of rules (case i). For case (ii), we can ensure that in $q$ there was a connected component $C_j$ that had an atom with all the variables in the newly introduced connected component $C'_i$, thus $\ccard{C'_i} \leq \ccard{C_j}$. 
We can then ensure that the generated \UCQ{}-rewritings have a
bounded connected cardinality. Therefore, \cdr{} rules can only produce a finite number \UCQ{}-rewritings (Lemma \ref{finite_cardk}). 
\end{proof}

Definition \ref{cdrdef} also applies to disjunctive existential rules. However, a rule generated by a disjunctive rewriting step involving a \cdr{} rule might not be a \cdr{} rule. Therefore, the \fus{} property cannot be extended to connected domain restricted disjunctive rules.

\begin{example}
\label{ex:cdrrwnotcdr}
Consider the rule $\atomx{a}{X}, \atomx{b}{Y} \rightarrow [r(X,Y), \atomx{c}{X}, \atomx{c}{Y}]$ and the \CQ{} $\atomx{r}{X,Y},\atomx{s}{X,Y}.$ They both generate a new disjunctive rule $\atomx{a}{X}, \atomx{b}{Y}, \atomx{s}{X,Y} \rightarrow [\atomx{c}{X}, \atomx{c}{Y}]$ that is not a \cdr{} rule. If another \CQ{} $\atomx{c}{X},\atomx{s}{X,Z}$ is used instead, then a disjunctive rule that is not a \cdr{} rule is again generated, i.e., $\atomx{a}{X}, \atomx{b}{Y}, \atomx{s}{X,Z} \rightarrow [\atomx{r}{X,Y}, \atomx{c}{Y}]$.
\end{example}

We use a similar approach to define a new rule class based on linear rules.

\begin{definition}[Connected linear rule] 
\label{clrdef}
A rule is called \emph{connected linear} rule (\clr{}) if every atom in the head either does not contain variables from the body or contains variables from only one connected component in the body and this connected component has only one atom.

\end{definition}

Both rules of Example \ref{ex:cdr} are also connected linear rules. 

\begin{example}
\label{ex:conlin}
The following rule is not a \cdr{} but it is clearly a connected linear rule. 
\begin{equation*}
\atomx{graduated}{X,Z}, \atomx{graduated}{Y,W} \rightarrow \atomx{exam}{V}, \atomx{passed}{X,V}, \atomx{passed}{Y,V}
\end{equation*}

\end{example}

\begin{theorem}
\label{clrfus}
A set of connected linear existential rules is a \fus{}.
\end{theorem}
\begin{proof} 
A \UCQ{}-rewriting $q'$ that is generated using a \clr{} rule $r$ and a \CQ{} $q$ has new atoms $a'_i$ that are either (i) not connected to the rest of the query or (ii) only connected to variables which were already present in an atom of $q$. 

A \clr{} prevents an atom in the head of the rule from containing variables from two different atoms (or connected components) in the body. Therefore, the rewritten atoms in $q$ are never replaced by more than one corresponding atom that is connected to the rest of the query. This ensures that the newly formed connected component in $q'$ will not have more atoms than those  existing in $q$. The rewriting $q'$ can have other atoms that are not a ``replacement'' of atoms in $q$ but those atoms are not connected to the atoms that existed in $q$. They come from other connected components that were present in the body of the rules. Thus, the \UCQ{}-rewritings which are introduced using connected linear rules have a bound on the number of atoms in their connected components. Therefore, connected linear rules may only produce a finite number \UCQ{}-rewritings (Lemma \ref{finite_cardk_edges}).
\end{proof}

The definition \ref{clrdef} may also be extended to disjunctive existential rules. However, a rule generated by a disjunctive rewriting step involving a \clr{} rule might not be a \clr{} rule. Therefore, the \fus{} property cannot be extended to connected linear disjunctive rules.

\begin{example}
\label{ex:clrrwnotclr}
Consider a connected linear rule $\atomx{a}{X}, \atomx{b}{Y} \rightarrow [r(X,W), \atomx{c}{X}, \atomx{c}{Y}]$ and a \CQ{} $\atomx{c}{X},\atomx{s}{X,Z}$. We can generate new disjunctive rule (i.e., ${\atomx{a}{X}, \atomx{b}{Y}, \atomx{s}{X,Z} \rightarrow [\atomx{r}{X,W}, \atomx{c}{Y}]}$) that is not a connected linear rule.
\end{example}

A disjunctive existential rule can be restricted to have disconnected disjoints. 

\begin{definition}[disconnected disjunction]\label{discdisj}
    A disjunctive existential rule has \emph{disconnected disjunction} if the disjoint components in the head of the rule never share variables with the same connected component in the body of the rule. A disjunctive existential rule that has disconnected disjunction is called a \emph{D-disjunctive existential rule} or \DDER{}.
\end{definition}

\begin{theorem}
     \label{dderrw}
The rewritings of \DDER{}s are also \DDER{}s.    
\end{theorem}
\begin{proof}
    Let $C_1,\ldots,C_n \rightarrow [D_1,\ldots,D_m]$ be a \DDER{} $r_1$. Without loss of generality, we define a rewriting $r_2$ that removes $D_1$ and introduces new atoms $B$ in the body, i.e., $B, C_1,\ldots,C_n \rightarrow [D_2,\ldots,D_m]$. The atoms in $B$ may possibly merge some connected components $C_i$ of the body of the rule. In particular, those that were connected to $D_1$. However, those components cannot be connected to any of the remaining disjoints $[D_2,\ldots,D_m]$ due to the fact that $r_1$ is a \DDER{}. Thus, $r_2$ is also a \DDER{}.
\end{proof}

Using similar reasoning, we can also affirm that \cdr{} (\clr{}) that are also \DDER{}, generate rewritings that are also \cdr{} (\clr{}).

\begin{theorem}
\label{disc-stops}
Let \Rules{} be a \DDER{} that is also a \cdr{} (\clr{}). Then, $\Rules{}$ is also a \fus{}.
\end{theorem}
\begin{proof}
We state that the Algorithm \ref{rwgen} cannot generate infinitely many rewritings if the rules are \DDER{} and \cdr{} (\clr{}).

Let $Q$ be a \UCQ{} and $M$ be the maximum cardinality (width) of the bodies in the rules of $\Rules$ and the \CQ{}s in $Q$. Given that all the rules in $\Rules{}$ are \cdr{} (\clr{}), a rewriting step will only produce queries with a cardinality bounded by $M$. Additionally, the cardinality of the bodies of rules produced as rewritings of disjunctive rules in $\disjunctive{\Rules{}}$ will be bounded by $M$ because they are \DDER{}. The newly generated existential rules will have the same \fus{} property of $\existential{\Rules{}}$, i.e., \cdr{} (\clr{}) and this ensures that at every step of the algorithm $\existential{\Rules{}}$ is a \fus{}.

Thus, the rewriting Algorithm \ref{rwgen} will stop due to the fact that it can only produce finitely many rewritings of the initial arguments.
\end{proof}

\section{System Description and Evaluation}\label{eval}

\completo{}\footnote{\url{http://image.ntua.gr/~gardero/completo3.0/}} is a query rewriting system that focuses on answering \NUCQ{}s in the framework of disjunctive existential rules. The system is implemented in java. The first version of \completo{} \cite{Alfonso2017} answers \NCQ{} using a reso\-lution-based approach to eliminate negated atoms. The proposed algorithm is complete only for a restricted type of queries. 

In the second version of the system \cite{Alfonso2018}, only queries with one negated atom are answered by being transformed  into rules. The approach is complete but  termination is guaranteed only when the resulting set of rules is a \fus{}.

The 3rd version of \completo{}\cite{semantic}  implements Algorithm \ref{rwgen} with deterministic one-step rewriting functions and answers queries with answer variables that have an arbitrary number of negated atoms. Algorithm~\ref{rwgen} can be seen as a generalization of both algorithms proposed  in \cite{Alfonso2017, Alfonso2018}. Indeed, queries with one negated atom are transformed into rules, while the rewriting defined for disjunctive rules is similar to what was presented in \cite{Alfonso2017} as constraint resolution.
Furthermore, \completo{} v3 takes advantage of the termination results for  knowledge bases consisting of a \fus{} and  \NUCQ{}s whose frontier is part of the answer variables of the query (Theorem 3.8 on \cite{semantic}), as well as for knowledge bases consisting only of linear elements (Theorem 3.9 on \cite{semantic}). Choosing $k=\infty$ allows the rewriting with respect to existential rules to be performed by an external rewriter if the are no answer variables in the queries. 

We present the latest version of the system called  \ecompleto{}\footnote{\url{https://github.com/gardero/ecompleto}} and it is implemented in Elixir. The system answers queries with answer variables that contain an arbitrary number of negated atoms with respect to disjunctive existential rules. Ontologies are provided in \dlgpp{} format, a proposed extension of \dlgp{}\footnote{\url{https://graphik-team.github.io/graal/papers/datalog+_v2.0_en.pdf}} v2.0 that allows the specification of disjunctive existential rules and negated atoms in queries.

Disjunction in  \dlgpp{} is specified in the head of a rule by writing a list of the disjoints enclosed in squared brackets. The disjoint elements can be a single atom or several atoms enclosed in brackets, e.g.,
\vspace{3pt}
\begin{verbatim}
[disj. rule] [leaf(X), (inner_node(X), edge(X,Y))] :- node(X).
\end{verbatim}
\vspace{3pt}
Negation in queries with negated atoms is specified with the minus symbol before an atom, e.g.,
\vspace{3pt}
\begin{verbatim}
[q neg] ? :- person(X), -marriedTo(X,Y). 
\end{verbatim}
\vspace{3pt}
\subsection{Experiments}

To the best of our knowledge, there is no other system that produces \UCQ{}-rewritings for \NUCQ{}s with universally quantified negation with respect to disjunctive existential rules. Therefore, the experiments were performed in order to assess the performance of \ecompleto{} producing \UCQ{}-rewritings. We used an Intel(R) Core(TM) i5-7300HQ CPU at 2.50 GHz with 24 GB of \RAM{} running  Ubuntu 22.04.

For the experiments, we used two ontologies that contain negative constraints and have been used in previous research papers based on queries with negation \cite{Alfonso2017, Alfonso2018}. The first is the Lehigh University Benchmark (\lubm{}) ontology \cite{lubm}, enriched with 70 additional disjoint classes axioms added for the atomic \emph{sibling} classes, i.e., for classes asserted to share the same super-class. 
Secondly, we used the \travel{} ontology\footnote{\url{https://protege.stanford.edu/ontologies/travel.owl}} that contains 10 disjoint class axioms. The OWL 2 ER \cite{owl2er} fragment of both ontologies was translated into existential rules. We were not able to prove the \fus{} property for the set of existential rules obtained from neither of the two ontologies we used. 

We prepared 500 \NCQa{}s for each ontology and \ecompleto{}  produced finite \UCQ{}-rewritings for all the \NCQa{}. Each of the queries contains three atoms and two of them are negated. Each query contains one variable in the frontier, which is also an answer variable. The queries were generated by performing Association Rule Mining \cite{ar} on a dataset obtained from the assertions of the ontologies. The ontologies and the queries used are publicly available \footnote{\url{http://image.ntua.gr/~gardero/completo3.0/ontologies/}}.

\begin{table}[h]
\begin{center}
\begin{minipage}{\textwidth}
\caption{Distribution metrics computed on the query rewriting runtime and the memory used for both ontologies.}\label{tabMetrics}
\begin{tabular*}{\textwidth}{@{\extracolsep{\fill}}lcccc@{\extracolsep{\fill}}}
\toprule%
& \multicolumn{2}{@{}c@{}}{LUBM} & \multicolumn{2}{@{}c@{}}{Travel} \\
\cmidrule{2-3}\cmidrule{4-5}%
Metric & Runtime (m) & Memory (Mb) & Runtime (m) & Memory (Mb)\\
\midrule

mean & 18.59 & 370.46 & 0.035 & 104.54 \\
std & 1.67 & 37.00 & 0.150 & 16.70\\
min & 15.33 & 328.00 & 0.020 & 93.00\\
25\% & 17.62 & 350.00 & 0.022 & 101.00\\
50\% & 18.43 & 359.00 & 0.023 & 103.00\\
75\% & 19.10 & 371.00 & 0.024 & 105.00\\
max & 24.76 & 526.00 & 2.454 & 337.00\\

\botrule
\end{tabular*}
%\footnotetext{Note: This is an example of table footnote. This is an example of table footnote this is an example of table footnote this is an example of~table footnote this is an example of table footnote.}
\end{minipage}
\end{center}
\end{table}

Table \ref{tabMetrics} shows the mean, std, min, max, and the 25th, 50th, and 75th percentiles of the \UCQ{} rewriting runtime and the used \RAM{} memory for both ontologies. The \UCQ{} rewriting runtime is on average 500 times faster for the \travel{} ontology than for the \lubm{} ontology. The rewriting process for the \travel{} ontology uses on average one third of the \RAM{} memory used to rewrite queries compared to the \lubm{} ontology.

\begin{figure}[h]%
\centering
\includegraphics[width=0.8\textwidth]{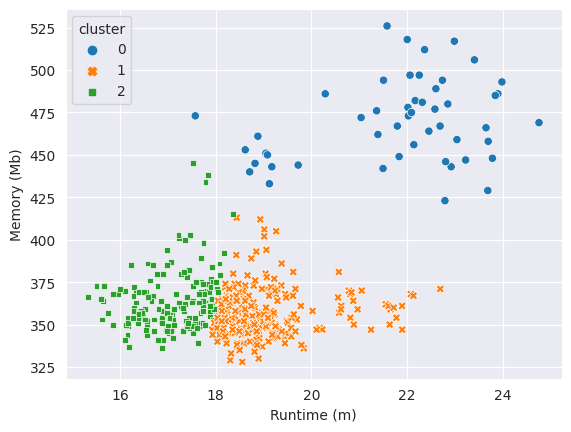}
\caption{Clustering of query rewriting runtime vs memory usage for LUBM.}\label{figCluster}
\end{figure}

\begin{figure}[h]%
\centering
\includegraphics[width=0.8\textwidth]{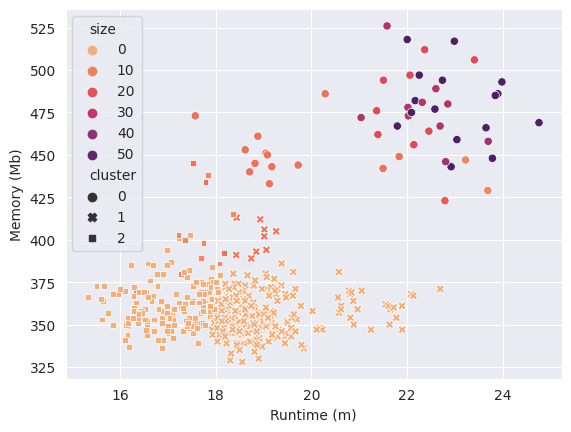}
\caption{Clustering of query rewriting runtime vs memory usage and size of the rewriting for LUBM.}\label{figClusterSize}
\end{figure}

Figure \ref{figCluster} shows the runtime vs \RAM{} memory of the rewriting process for each query of the \lubm{} ontology. There are 3 clusters that group the points according to their standardized coordinates. Figure \ref{figClusterSize} shows also the size of the \UCQ{} rewriting using colors. The darker the datapoint is, the larger the size of the corresponding rewriting is. The cluster grouping shows some correlation with the size of the rewriting.

\begin{table}[h]
\begin{center}
\begin{minipage}{\textwidth}
\caption{Distribution metrics computed on the query rewriting runtime and the memory used for both clusters in LUBM ontology.}\label{tabClusters}
\begin{tabular*}{\textwidth}{@{\extracolsep{\fill}}lccccc@{\extracolsep{\fill}}}
\toprule%
& \multicolumn{3}{@{}c@{}}{Runtime (m)} & \multicolumn{2}{@{}c@{}}{Memory (Mb)} \\
\cmidrule{2-4}\cmidrule{5-6}%
Metric & Cluster 0 & Cluster 1 & Cluster 2 & Cluster 0 & Cluster 1,2 \\
\midrule

count & 50 & 280 & 170 & 50 & 450\\
mean & 21.82 & 18.95 & 17.06 & 469.64 & 359.44\\
std & 1.72 & 0.88 & 0.68 & 24.31 & 15.49\\
min & 17.57 & 17.92 & 15.33 & 423.00 & 328.00\\
25\% & 21.37 & 18.40 & 16.53 & 449.25 & 349.00\\
50\% & 22.15 & 18.68 & 17.10 & 468.00 & 357.00\\
75\% & 22.90 & 19.20 & 17.65 & 485.75 & 367.00\\
max & 24.76 & 22.69 & 18.36 & 526.00 & 445.00\\
\botrule
\end{tabular*}
%\footnotetext{Note: This is an example of table footnote. This is an example of table footnote this is an example of table footnote this is an example of~table footnote this is an example of table footnote.}
\end{minipage}
\end{center}
\end{table}

Table \ref{tabClusters} shows the count, mean, std, min, max, and the 25th, 50th, and 75th percentiles of the \UCQ{} rewriting runtime and the used \RAM{} memory for each of the clusters of \lubm{} queries.

%\begin{figure}[h]%
%\centering
%\includegraphics[width=0.9\textwidth]{figures/clusters-join-lubm.png}
%\caption{Clustering of query rewriting runtime vs memory usage for LUBM.}\label{fig1}
%\end{figure}

\begin{figure}[h]%
\centering
\includegraphics[width=0.8\textwidth]{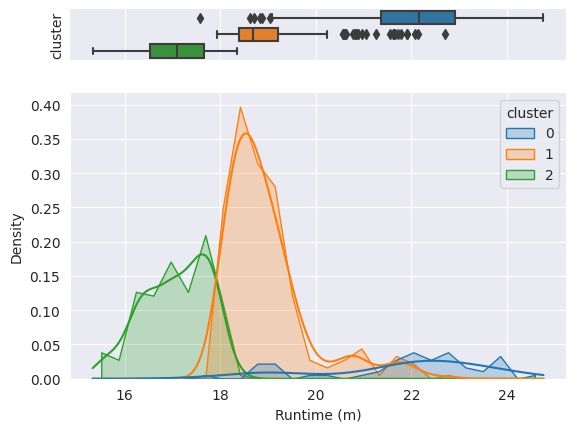}
\caption{Histogram of query rewriting runtime for LUBM.}\label{figRuntimeLubm}
\end{figure}

Figure \ref{figRuntimeLubm} shows the distribution of the query rewriting runtime for the \lubm{} ontology. The distribution is multimodal and it is split according to the cluster group.

\begin{figure}[h]%
\centering
\includegraphics[width=0.8\textwidth]{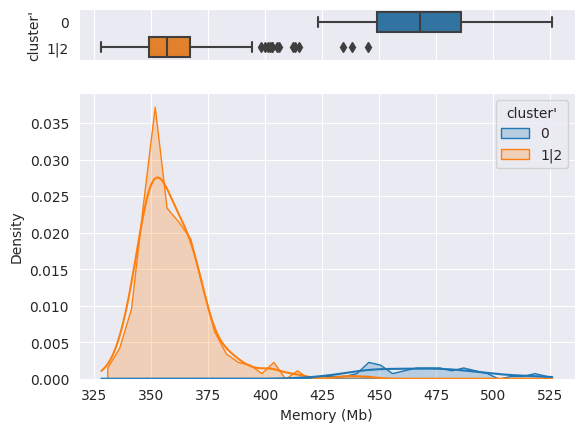}
\caption{Memory usage histogram for LUBM.}\label{figMemLubm}
\end{figure}

The distribution of the \RAM{} memory used in the rewriting process is shown in Figure \ref{figMemLubm}. We can notice a bimodal shape, despite having 3 clusters that group the queries.

%begin{figure}[h]%
%\centering
%\includegraphics[width=0.8\textwidth]{figures/hist-travel-mem.png}
%\caption{Memory usage histogram for Travel.}\label{fig1}
%\end{figure}

\section{Discussion}\label{discussion}

The experiments conducted in this study aimed to evaluate the performance of the \ecompleto{} system in producing UCQ-rewritings for NUCQs with respect to disjunctive existential rules. These experiments provide valuable insights into the efficiency and resource requirements of \ecompleto{} when handling complex queries and ontologies.

Various performance metrics were computed to evaluate \ecompleto{}'s performance in terms of query rewriting runtime and memory usage. These metrics included mean, standard deviation, minimum, maximum, and percentiles (25th, 50th, and 75th) for both runtime and memory consumption.

The results demonstrated a significant difference in query rewriting runtime between the LUBM and Travel ontologies. On average, the rewriting process for the Travel ontology was approximately 500 times faster than that for the LUBM ontology.

In terms of memory usage, the Travel ontology exhibited greater efficiency, utilizing only about one-third of the RAM memory required for query rewriting compared to the LUBM ontology.

Clustering analysis revealed that the queries formed distinct clusters based on their runtime and memory consumption. The size of the UCQ rewriting also demonstrated a degree of correlation with the cluster grouping.

The distribution of query rewriting runtime for the LUBM ontology was found to be multimodal, with clusters representing different query characteristics. Similarly, the distribution of RAM memory usage exhibited a bimodal pattern despite the presence of three query clusters.

\section{Conclusion}\label{conclusion}

In conclusion, this paper has investigated the domain of conjunctive query rewriting, an important technique used in the context of query answering within rule-based systems. The main objective of conjunctive query rewriting is to transform an input query into a complete rewriting that does not need the rules to represent all the possible responses. While this approach is very useful in dynamic data environments characterized by stable queries and rules, cases featuring static data combined with dynamically varying queries need a forward chaining approach.

The focus of our investigation is dedicated to the domain of conjunctive query answering, specifically on queries that may have universally quantified negation, within the framework of disjunctive existential rules. We have taken on the computational challenge involved in determining the existence of finite and complete \UCQ{}-rewritings, as well as the identification of finite unification sets (\fus{}) of rules. Our contributions include the introduction of two novel rule classes: \emph{connected linear rules} and \emph{connected domain restricted rules}, which have the \fus{} property and are more expressive than their antecedent rule classes, namely linear rules and domain-restricted rules.

Furthermore, we have introduced the concept of \emph{disconnected disjunction} in the framework of disjunctive existential rules. This concept allows us to achieve the \fus{} property for \emph{connected linear rules} and \emph{connected domain restricted rules} also in the presence of disjunctive rules with disconnected disjunction. 

In terms of practical implementation, we introduced our system, \ecompleto{}, specifically designed for the task of query rewriting within the framework of disjunctive existential rules. The system handles  \NUCQ{} queries that may include universally quantified negation. In addition, we expanded the \dlgpp{} format to facilitate the specification of disjunctive existential rules and the inclusion of negated atoms within queries.

The empirical evaluation of our system included a series of experiments conducted on established ontologies, namely the Lehigh University Benchmark \lubm{} and \travel{}, both augmented with supplementary axioms. The outcome of these experiments showed the consistent performance of \ecompleto{} in generating finite \UCQ{}-rewritings for a diverse set of queries. Furthermore, the system exhibited acceptable efficiency during the rewriting process.

Finally, the experiments provided valuable insights into \ecompleto{}'s performance when producing \UCQ{}-rewritings for \NUCQ{}s with universally quantified negation and disjunctive existential rules. The significant differences in runtime and memory consumption between the \lubm{} and \travel{} ontologies emphasize the importance of considering ontology complexity when assessing \ecompleto{}'s performance. Clustering and distribution patterns offered additional insights into the performance of \ecompleto{} under different input queries. Overall, these findings contributed to the understanding of our system's capabilities and provided valuable information for researchers working with complex queries and ontologies in the context of disjunctive existential rules.

%\section{Conclusion}\label{sec13}

%Conclusions may be used to restate your hypothesis or research question, restate your major findings, explain the relevance and the added value of your work, highlight any limitations of your study, describe future directions for research and recommendations. 

%In some disciplines use of Discussion or 'Conclusion' is interchangeable. It is not mandatory to use both. Please refer to Journal-level guidance for any specific requirements. 

\backmatter

%%===========================================================================================%%
%% If you are submitting to one of the Nature Portfolio journals, using the eJP submission   %%
%% system, please include the references within the manuscript file itself. You may do this  %%
%% by copying the reference list from your .bbl file, paste it into the main manuscript .tex %%
%% file, and delete the associated \verb+\bibliography+ commands.                            %%
%%===========================================================================================%%

\bibliography{references}% common bib file
%% if required, the content of .bbl file can be included here once bbl is generated
%%\input sn-article.bbl

%% Default %%
%%\input sn-sample-bib.tex%

\end{document}

%% file: defs.tex
\newcommand{\CalD}{\mathcal{D}}

\newcommand{\CalQ}{\mathcal{Q}}
\newcommand{\CalR}{\mathcal{R}}

\newcommand{\entailsp}{\models_{?}}

\newcommand{\comp}[1]{\bar{#1}} %% literal complement

\newcommand{\atomx}[2]{\textit{#1}({{#2}})}

\newcommand{\ccard}[1]{\textit{card}^*({{#1}})}

\newcommand{\cwidth}[1]{\textit{width}^*({{#1}})}

\newcommand{\pred}[2]{\textit{#1}/{{#2}}}

\newcommand{\vvector}[1]{{\bf{#1}}}

\newcommand{\NCQ}{\ensuremath{\mathsf{CQ}^{\neg}}}
\newcommand{\NUCQ}{\ensuremath{\mathsf{UCQ}^{\neg}}}

\newcommand{\CQ}{\ensuremath{\mathsf{CQ}}}
\newcommand{\UCQ}{\ensuremath{\mathsf{UCQ}}}

\newcommand{\DSF}{\ensuremath{\mathsf{DSF}}}
\newcommand{\CSF}{\ensuremath{\mathsf{CSF}}}

\newcommand{\existential}[1]{{#1}^\exists}
\newcommand{\disjunctive}[1]{{#1}^\vee{}}
\newcommand{\constraints}[1]{{#1}^\bot}

\newcommand{\Rules}{\ensuremath{\mathcal{R}}}

\newcommand{\size}[1]{\ensuremath{\vert{#1}\vert}}

\newcommand{\DDER}{\ensuremath{\mathsf{DDER}}}

\newcommand{\NCQa}{\ensuremath{\mathsf{CQa}^{\neg}}}

\newcommand{\fus}{\ensuremath{\textit{fus}}}
\newcommand{\dr}{\ensuremath{\textit{dr}}}
\newcommand{\cdr}{\ensuremath{\textit{cdr}}}
\newcommand{\clr}{\ensuremath{\textit{clr}}}

\newcommand{\completo}{\textsc{Completo}}
\newcommand{\ecompleto}{\textsc{ECompleto}}
\newcommand{\lubm}{\textsc{LUBM}}
\newcommand{\travel}{\textsc{Travel}}
\newcommand{\RAM}{\textsc{RAM}}
\newcommand{\dlgpp}{\textsc{DLGP+}}
\newcommand{\dlgp}{\textsc{DLGP}}